\documentclass[11pt]{amsart}
\usepackage{fullpage}
\usepackage{mydefs}
\usepackage{amssymb}
\usepackage{graphicx}

\usepackage{amsmath}

\usepackage{graphicx}
\usepackage{amssymb}
\usepackage{multicol}
\usepackage{epstopdf}
\usepackage{mydefs}
\usepackage{algorithm}
\usepackage{algorithmic}

\usepackage{url}

\newcommand{\alg}[1]{\textbf{#1}}   

\newcommand{\calA}{\mathcal{A}}      

\title{A Fully Distributed Algorithm for Throughput Performance in Wireless Networks}

\author[E. I.\'Asgeirsson]{Eyj\'olfur I. \'Asgeirsson}
\address[M. M. Halld\'orsson]{School of Science and Engineering\\
 Reykjavik University\\
 101 Reykjavik, Iceland\\}
\email{eyjo@ru.is}

\author[M. M. Halld\'orsson]{Magn\'us M. Halld\'orsson}
\address[M. M. Halld\'orsson]{School of Computer Science\\
 Reykjavik University\\
 101 Reykjavik, Iceland\\}
\email{mmh@ru.is}

\author[P. Mitra]{Pradipta Mitra}
\address[P. Mitra]{School of Computer Science\\
Reykjavik University\\
Reykjavik 101, Iceland}
\email{ppmitra@gmail.com}

\begin{document}

\begin{abstract}
We study link scheduling in wireless networks under stochastic arrival processes
of packets, and give an algorithm that achieves stability in the
physical (SINR) interference model. The efficiency of such an algorithm is the
fraction of the maximum feasible traffic that the algorithm can handle
without queues growing indefinitely.  
Our algorithm achieves two important goals:
(i) efficiency is independent of the size of the network, and (ii) the
algorithm is fully distributed, i.e., individual nodes need no
information about the overall network topology, not even local
information.
%
\end{abstract}

\maketitle 
\section{Introduction}
Designing high-performance scheduling algorithms for wireless networks has become
an increasingly important topic in recent years. Scheduling in a wireless environment is a non-trivial problem,
since simultaneous transmissions interfere with each other in complex ways. 
A two-fold
challenge of appropriately modelling the interference, and then developing algorithms
for that model presents itself. 
Furthermore, in many realistic settings, a centralized controller
cannot be assumed, and algorithms that work in a distributed fashion have to be developed.

In this work, we are interested in stability and the associated
throughput performance of scheduling algorithms for wireless networks under realistic interference models.
We assume that packets arrive at potential senders according to a stochastic process, and the goal
of an algorithm is to schedule these transmissions so that the queues of unscheduled packets at
each sender remain bounded (in which case, the system is called \emph{stable}). 
A rich body of research has been devoted to dealing with this issue in a variety of settings. 
The seminal work of Tassiulas and Ephremides \cite{TE92} established
that an optimal scheduling policy exists,
one that 
stabilizes the system under all arrival rates for which stability is potentially possible.
In most settings, however, such a ``perfect'' solution is
computationally intractable, and additionally a distributed implementation is unlikely.

Hence, the search for efficient and/or distributed algorithms which, if not as good as 
the optimal algorithm, are nevertheless useful.
 Since these algorithms may not
stabilize all feasible arrival processes, 
the concept of \emph{efficiency} of an algorithm has been introduced, 
being the fraction of  $\Lambda$  that the algorithm can stabilize,
where $\Lambda$ is the space of
arrival processes that the optimum algorithm can stabilize. 
There have been many approaches to developing such algorithms. 
A natural step in the search for efficient algorithms is to seek
maximal solutions.
In the context of wireless networks this is known as Greedy Maximal Scheduling (GMS) algorithm
\cite{bestInfocom08} or Longest Queue First (LQF) algorithm
\cite{secordorder}.
The stability and efficiency of LQF has been investigated extensively \cite{bestInfocom08,DBLP:conf/mobihoc/LiBX09,secordorder}.
Many other approaches have been proposed as well 
\cite{DBLP:conf/sigmetrics/ModianoSZ06,DBLP:journals/ton/JooS09,shahfocs11}. 

Most analytic work on wireless networks has been done in 
graph-based interference models
(e.g.\ \cite{DBLP:conf/sigmetrics/ModianoSZ06,bestInfocom08,DBLP:conf/mobihoc/LiBX09,secordorder,shahfocs11}).
In these models, wireless links (a link is a sender-receiver pair) that are neighbors in a specified link-graph cannot transmit simultaneously.
Though interesting in their own right, these models are known to over-simplify interference
coupling~\cite{MaheshwariJD08,Moscibroda2006Protocol}. As a result,
many research communities working on wireless networks have increased their
focus on the so-called \emph{physical model} or the SINR model. In this model,
a transmission is considered successful if the signal received at the intended receiver
is suitably larger than the cumulative interference due to all other transmissions in the network, plus
the ambient noise.  In the SINR model, solving the characterization in \cite{TE92} is equivalent to solving the maximum
weighted \emph{capacity} problem, which is known to be NP-hard \cite{gouss2007}
and additionally has no known constant factor approximation algorithm (quite apart from the issue
of distributed implementation).

In this paper, we develop an algorithm that is completely distributed, with nodes requiring no topological information about the network (not even information about ``neighbors''), and achieves an efficiency ratio that is independent of the network size (i.e., the number of links).
Thus, the algorithm is scalable in relation to network size. It  can also operate in an asynchronous setting, with nodes appearing arbitrarily (as long as the stochastic process meets the required condition).

We give simulation results which lend credence to the theoretical bounds.
Our algorithm is extremely robust under fairly high amount of load (achieving efficiency ratios bordering on $0.5$).

The only other work (that we are aware of) on stability of algorithms in
the SINR model is the recent work by Le \emph{et.\ al.} \cite{lqfmobihoc},
who analyze the stability properties of the LQF algorithm.
The authors show that the basic
LQF is not efficient, but a variation of it that localizes interference is shown to work, with  
efficiency similar to ours. In a related work, \cite{lee09} also consider the SINR model, but the model
there is different (links are always feasible, but have different data rates based on the SINR achieved. This sort of problem is rather different from the ``combinatorial'' situation at hand).

A distinguishing feature of our distributed implementations is that they require
almost
no additional ``infrastructure''. Often distributed algorithms for wireless networks have to
assume another underlying information infrastructure that can be used to run a localized and/or
distributed algorithm, and that infrastructure, moreover, is not subject to the interference constraints
of the original network. This is the case with \cite{lqfmobihoc}, as well as many other works on the topic (\cite{DBLP:conf/sigmetrics/ModianoSZ06,lee09} for example). 
This is a rather strong assumption,
especially in light of the fact that in a wireless network, one is usually trying to establish
such an infrastructure in the first place.
It is interesting that we can do without them while obtaining high throughput performance.

From a technical perspective, we adopt the vocabulary and techniques developed in the context of worst-case algorithmic research on the SINR model (\cite{MoWa06,HW09,KesselheimSoda11,SODA11}). The concept of ``affectance'' (defined later) developed in some of these works turns out to be
quite effective in this context. This approach may have further applications in the study of stability
of wireless networks.

The paper is organized as follows.
In Section \ref{sec:lngthclass}, we describe our algorithm and state
the main result.
In Section \ref{sec:model} we present the system model,
and discuss related work further in Sec.~\ref{sec:results}.
The proof of the stability result is given in Section \ref{sec:proof}.
Finally, in Section \ref{sec:simulations} we present simulation results.

\section{Algorithm and Result}
\label{sec:lngthclass}
The wireless network is modeled as a set $L$ of $n$ links, where
each link $l_u \in L$ represents a potential transmission from a sender
$s_u$ to a receiver $r_u$, each a point in a metric space. 

We assume that packets arrive at the sender of each link $l_u$ according to a stochastic
process with average arrival rate $m_u$.

The  extremely simple and fully distributed algorithm is as follows.

\begin{algorithm}                      
\caption{Reflect (Run by each link $l_u$ in the system)}          
\label{alg2}                           
\begin{algorithmic}[1]                    
     \STATE $Q_u \leftarrow \emptyset$ (queue of outstanding packets)
     \FOR{$t \leftarrow{} 1, 2, \ldots $}
       \STATE Let $\calA$ be the set of packets that arrive at the beginning of time slot $t$
       \STATE Add $\calA$ to the end of $Q_u$
       \IF{$Q_u$ is non-empty}
         \STATE Transmit a packet from $Q_u$ with probability $2.5 \cdot m_u$\label{trline}
       \ENDIF
       \ENDFOR
\end{algorithmic}
\label{alg2fig}
\end{algorithm}

Our main result is:
\begin{theorem}
\label{th:lngthclass}
For all given networks with links on metric spaces,  
and all sublinear, length-monotone  power assignments,
\alg{Reflect} achieves an efficiency ratio independent of $n$.
\end{theorem}

\section{Some Preliminaries}
\label{sec:model}
The distance between two points $x$ and $y$ is denoted $d(x,y)$.
The distance from
$l_u$'s sender to $l_v$'s receiver is denoted $d_{uv} = d(s_{u}, r_{v})$.
The length of link $l_u$ is denoted 
simply by $\ell_u = d(s_u, r_u)$. The link set is associated with a \emph{power assignment} $P$, which is an assignment of a transmission
power $P_u$ to be used by the sender of each link $l_u \in L$.
The signal received at point
$y$ from a sender at point $x$ with power $P_x$  is $P_x/d(x, y)^\alpha$ where the constant 
$\alpha > 0$ is the
\emph{path-loss exponent}. 

We can now describe the  \emph{physical} or SINR-model of interference. In this model, a receiver $r_u$
successfully receives a message from the sender $s_u$ if and only if the
following condition holds:
\begin{equation}
 \frac{P_u/\ell_u^\alpha}{\sum_{l_v \in S \setminus  \{l_u\}}
   P_{v}/d_{vu}^\alpha + N} \ge \beta \ , 
 \label{eq:sinr}
\end{equation}
where $N$ is the environmental noise, the constant $\beta$ denotes the minimum
SINR (signal-to-interference-noise-ratio) required for a message to be successfully received,
and $S$ is the set of concurrently scheduled links in the same \emph{slot} (we assume that time is slotted).
We say that $S$ is \emph{SINR-feasible} (or simply \emph{feasible}) if (\ref{eq:sinr}) is
satisfied for each link in $S$.

A power assignment $P$ is \emph{length-monotone} if $P_v \ge P_w$ whenever
$\ell_v \ge \ell_w$ and \emph{sublinear} if $\frac{P_v}{\ell_v^{\alpha}} \le \frac{P_w}{\ell_w^{\alpha}}$ whenever $\ell_v \ge \ell_w$. This class includes the most interesting and practical power assignments, such as uniform power (all links
use the same power), linear power ($P_u = \ell_u^{\alpha}$, known to be energy efficient in the presence of noise), and
mean power ($P_u = \ell_u^{\alpha/2}$, the assignment that produces maximum capacity in this class \cite{SODA11}). 
Let $\Delta = \frac{\ell_{\max}}{\ell_{\min}}$ where $\ell_{\max}$ and $\ell_{\min}$ are, respectively, the maximum and minimum lengths of links in $L$.

\begin{defn}
The \emph{affectance} $a^P_{v}(u)$ of link $l_u$ caused by another link $l_v$,
with a given power assignment $P$,
is the interference of $l_v$ on $l_u$ relative to the signal
received, or
  \[ a^P_{v}(u) 
     = \min\left\{1, c_u \frac{P_{v}}{P_u} \cdot \left(\frac{\ell_u}{d_{vu}}\right)^\alpha\right\} \ ,
  \] 
where $c_u = \beta/(1 - \beta N \ell_u^\alpha/P_u)$. 
\end{defn}
We need the following assumption.
\begin{assumption}
$c_u \leq 2 \beta$ for any link $l_u$.
\end{assumption}
This is is fairly reasonable assumption. It simply says that in the
absence of other links, the transmission succeeds comfortably. The
constant $2$ is not fundamental; any value greater than 1 would suffice.

Since $c_u \geq \beta$ by definition, this implies that
\begin{equation}
\frac{c_v}{c_u} \leq 2 \text{ for any two links } l_u, l_v
\label{eqn:cvassumption}
\end{equation}

The definition of affectance was introduced in \cite{GHWW09,HW09} and achieved the form
we use in \cite{KV10}.
When clear from the context we drop the superscript $P$. 
Also, let $a^P_v(v) = 0$.
Using affectance, Eqn.\ \ref{eq:sinr} can be rewritten as 
\begin{equation}
a^P_S(u) \equiv \sum_{l_v \in S}a^P_{v}(u) \leq 1\ ,
\end{equation}
for all $l_u \in S$.

\emph{Signal-strength and robustness.}
A \emph{$\delta$-signal} set of links is a set of links where the
affectance on any link is at most $1/\delta$.
A set is SINR-feasible iff it is a 1-signal set. We know:

\begin{lemma}[\cite{esatalg}]
Let $\ell_u, \ell_v$ be links in a $q^\alpha$-signal set.
Then, $d_{uv} \cdot d_{vu} \ge q^2 \cdot \ell_u \ell_v$. 
\label{lem:ind-separation}
\end{lemma}

Now we switch to the aspects of this paper related to queueing theory and stochastic processes.
We first define stability.
\begin{defn}
An algorithm \emph{stabilizes} a network for a particular arrival process if, under that
arrival process, the average queue size 
is bounded (at any given time).
\end{defn}
The \emph{throughput region} is then the set of all possible arrival rate
vectors such that there exists some scheduling policy that can stabilize the
network.

As proved in \cite{TE92}, the \emph{throughput region} is characterized by
$$\Lambda = \{\lambda : \lambda \preceq \phi, \text{for some } \phi \in Co(\Omega)\} \ ,$$
where $\Omega$ is the set of all maximal feasible schedules (meaning arrival processes 
that put weight $1$ on a single maximal feasible set) and $Co(\Omega)$ is the convex hull of $\Omega$.
Note that $\lambda$ and $\phi$ are $n$-dimensional vectors and $\lambda \preceq \phi$ means each element
of $\lambda$ is upper bounded by the corresponding one in $\phi$.

Since fast and/or distributed algorithms might not stabilize all of $\Lambda$, one hopes
to achieve a large \emph{efficiency ratio}.
\begin{defn}
The \emph{efficiency ratio} $\gamma$ of a scheduling algorithm 
is
$\gamma = \sup\{\eta: \text{all networks are}$ stabilized for all $\lambda \in \eta \Lambda \}$,
where $\eta \Lambda = \{\eta \lambda : \lambda \in \Lambda\}$.
\end{defn}
We assume that the arrival process on a link is i.i.d.\ across time, and different links are independent of each other. 

We will use $M_i$ to denote both maximal feasible sets, and characteristic vectors of said sets (the usage being clear from context).
For a given efficiency ratio $\gamma$, 
it must hold for all permissible arrival rate vectors $\lambda$ that
$\lambda \preceq \sum_{i} m_i M_i$, where and $m_i$ are weights such that 
\begin{equation}
\label{migamma}
\sum_i m_i  = \gamma \ .
\end{equation}
It can be easily seen that for any link $l_u$, 
\begin{equation}
m_u  = \sum_{i: l_u \in M_i} m_i \leq \gamma \ .
\label{mudef}
\end{equation}

\section{Related work}
\label{sec:results}


As stated in Thm. \ref{th:lngthclass}, the efficiency of the algorithm is independent of 
$n$ (the number of links in the system). It is, however, dependent on another network parameter $\Delta$, the ratio between the longest and the shortest link in the system (the proof in the next section contains the exact expression). The only comparable work
on this model \cite{lqfmobihoc} has the same dependence on $\Delta$ (this is not explicitly stated in the paper, but can be seen to be necessary). The main discriminating feature of our work is that it is distributed in a much stronger sense. The algorithm in \cite{lqfmobihoc} can be characterized as ``localized'', where each link needs to be aware of and have communicated with other links in its neighborhood. We have no need for such infrastructure.

In terms of efficiency ratio, 
a range of results have been derived
in a variety of models. Naturally one seeks efficiency of $1$ whenever
possible \cite{DBLP:conf/sigmetrics/ModianoSZ06}, but results for efficiency ratio of $1$ under certain conditions \cite{secordorder}, or $\frac{1}{2}$ \cite{DBLP:conf/infocom/DaiP00}, or $\frac{1}{6}$ \cite{bestInfocom08} can be found in the literature.
Ratios in terms of certain network characteristics are known as well -- such as in terms of the
degree of the interference graph \cite{DBLP:journals/tit/ChaporkarKLS08} or 
the local pooling factor \cite{bestInfocom08,lqfmobihoc}.
In \cite{mimosinr10}, the abstract SINR model (received signal is a general function instead of being length-based) in the context of MIMO networks is studied. An efficiency ratio based on a  system-specific value (``effective interference number'') is derived, with no direct comparison with distance-based SINR models.
 
For the SINR model, an efficiency ratio that is an ``unconditional'' constant (independent of both $n$ and $\Delta$)
is not known.

%
\section{Proof of Stability}
\label{sec:proof}
We now present a proof of Thm.~\ref{th:lngthclass}.

Note that the probability $2.5 \cdot m_u$ used for link $l_u$ in the algorithm is well-defined, since we claim stability with constant efficiency ratio bounded from above by $\frac{1}{3}$.


We first need the following observation.

\begin{observation}
For any two links $\ell_u$ and $\ell_v$ using a length monotone, sub-linear power assignment, 
$$\frac{P_u \ell_v^{\alpha}}{P_v \ell_u^{\alpha}}  \leq
\Delta^{\alpha}\ .$$
\label{lem:signalincr}
\end{observation}
\begin{proof}
If $\ell_v \le \ell_u$, it holds by sub-linearity that
$P_u/\ell_u^\alpha \le P_v / \ell_v^\alpha$.
Otherwise, if $\ell_u \le \ell_v$, then by monotonicity $P_u \le P_v$
and by definition of $\Delta$, $\ell_v^\alpha/\ell_u^\alpha \le \Delta^\alpha$.
\end{proof}

The following key lemma shows that no link is affected too much by any
single feasible set.

\begin{lemma}
Consider a feasible set $S$ and a link  $l_v$ (not necessarily a member of $S$).
Then,
\begin{equation}
\sum_{l_z \in S} a_{z}(v) \leq \kappa \cdot \Delta^{\alpha}\ ,
\end{equation}
for some constant $\kappa$.
\label{lem:load}
\end{lemma}

\begin{proof}
We use the signal strengthening technique of \cite{HW09}. For this, we decompose 
the set $S$ to $\lceil 2 \cdot 3^\alpha/\beta \rceil^2$ sets, each a $3^\alpha$-signal set.
We prove the claim for one such set; since there are only constantly
many such sets, the overall claim holds (with the appropriate increase
in the constant factor).
Let us reuse the notation $S$ to be such a $3^\alpha$-signal set.

Consider the link $l_u = (s_u, r_u) \in S$ such that $d(r_{v}, r_{u})$ is minimum.
Also consider the link $l_w = (s_w, r_w) \in S$ such that $d(s_{w}, r_{v})$ is minimum.
Let $D = d( r_v, r_{u})$. We claim that for all links $l_x = (s_x,
r_x) \in S$ with $\ell_x \neq \ell_w$, it holds that
\begin{equation}
d(s_{x}, r_v) \geq \frac{1}{2} D \ .
\label{eqn:dist1}
\end{equation}
To prove this, assume, for contradiction, that $d(s_{x}, r_v) < \frac{1}{2} D$. Then,  $d(s_{w}, r_v)  < \frac{1}{2} D$, by definition of $l_w$.
Now, again by the definition of $l_u$, $d(r_{x}, r_v) \geq D$ and $d(r_{w}, r_v) \geq D$. Thus $\ell_w \geq d(r_w, r_v) - d(r_v, s_w) > \frac{D}{2}$
and similarly $\ell_{x} > \frac{D}{2}$. On the other hand $d(s_{w}, s_{x}) < \frac{D}{2} + \frac{D}{2} = D$.
Now, $d_{w x} \cdot d_{x w} \leq (\ell_{w} + d(s_{w}, s_{x}))(\ell_{x} + d(s_w, s_{x})) <  (\ell_{w} + D)(\ell_{x} + D)
< 9 \ell_{w} \ell_x$, contradicting Lemma \ref{lem:ind-separation}. 

Now that we have proven Eqn~\ref{eqn:dist1}, 
by the triangle inequality, $d_{ x u}  = d(s_{x}, r_{u}) \leq d(s_{x}, r_v) + d(r_v, r_{u}) \leq 3 d(r_v, s_{x}) = 3 d_{ x v}$. 
Applying Obs.~\ref{lem:signalincr}, we see that
\[ \frac{a_{x}(v)}{a_{x}(u)} \leq \frac{c_v}{c_u} \cdot  \frac{P_u}{P_v} \cdot  \frac{d_{xu}^{\alpha}}{d_{xv}^{\alpha}}
\cdot  \frac{\ell_{v}^{\alpha}}{\ell_{u}^{\alpha}}
 \leq 2 \cdot 3^{\alpha} \frac{P_u \ell_v^{\alpha}}{P_v \ell_u^{\alpha}} 
 \le 2 \cdot 3^{\alpha} \Delta^\alpha \ .  \]
where $\frac{c_v}{c_u} \leq 2$ follows from Eqn. \ref{eqn:cvassumption}.
Finally, summing over all links in $S$,
\begin{align*}
a_S(v) & = \sum_{l_z \in S}  a_{z}(v)  = a_{w}(v) + \sum_{l_z  \in S \setminus \{l_w\}}  a_{z}(v) \\
 & \leq 1 + 2 \cdot (3\Delta)^\alpha \sum_{l_z \in S  \setminus \{l_w\}} a_{z}(u) \\
 & \leq 1 + 2 \cdot (3\Delta)^\alpha \ ,
\end{align*}
where we use  $a_{w}(v) \leq 1$ by the definition of affectance, and 
$\sum_{l_z \in S  \setminus \{l_w\}} a_{z}(u) \leq 1$ since $S$ is feasible and $l_u \in S$.

This completes the proof setting $\kappa = 3^{\alpha + 1}$.
\end{proof}

We turn to the proof of Thm.~\ref{th:lngthclass}.
\begin{proof}  
We claim an efficiency ratio of $1/(6 \kappa \Delta^{\alpha})$ where $\kappa$ is the constant from Lemma 
\ref{lem:load}. 
Thus, it is enough to prove stability for all stochastic processes for which the following holds:
\begin{equation}
\sum_{i : M_i} m_i \leq \frac{1}{6 \kappa \Delta^{\alpha}}\ .
\label{eqn:efficiencybound}
\end{equation}

Consider the affectance on any link $l_u$
during the execution of the algorithm in a single slot, and denote it by $a(u)$. This
can be computed as $a(u) = \sum_{l_v} X_{v} a_{v}(u)$,
where $X_{v}$ is a Bernoulli random variable which is $1$ iff link $l_v$ has a non-empty
queue and chooses to transmit during the same slot. Now,
\begin{align*}
\Ex(a(u)) & = \sum_{l_v} \Ex(X_{v}) a_{v}(u)  \overset{1}{\leq} \sum_{l_v} 2.5 m_{v} a_{v}(u) \\& \overset{2}{=} 2.5 \sum_i m_i \sum_{M_i \ni l_v} a_{v}(u) \\
  & \overset{3}{\leq} 2.5 \sum_i m_i \kappa \Delta^{\alpha} \overset{4}{\leq} \frac{2.5 \kappa \Delta^{\alpha}}{6 \kappa \Delta^{\alpha}} = \frac{5}{12}\ ,
\end{align*}
where explanations of the numbered (in)equalities are:
\begin{enumerate}
\item $\Ex(X_{v}) \leq 2.5 m_{v}$ by the description of the algorithm.
\item By Eqn.~\ref{mudef} and rearrangement.
\item By Lemma \ref{lem:load}.
\item By Eqn.~\ref{eqn:efficiencybound}.
\end{enumerate}

Thus, with probability at least $\frac{1}{2}$, $a(u) < 1$ (by Markov's inequality). 
Hence, if $l_u$ has a non-empty queue, with probability at least
$\frac{5}{4}  \cdot m_u$, the queue size decreases. Note that the probability 
can potentially be higher than
$\frac{5}{4} \cdot m_u$, but never smaller. Therefore, 
this system is at least as efficient as the system
where in each slot the queue size reduces by $1$ with i.i.d.\ probability exactly $\frac{5}{4} \cdot m_u$.

The queue dynamics on a single link become equivalent to the following single server system 
with slotted time and an infinite queue. In this system, at the beginning of each time slot, $A$ packets arrive, where $A$ is a random variable on the non-negative
integers, with $\Ex(A) = m_u$. At the end of each slot, the server processes $D$ packets 
(or empties the queue), where
$D$ is a Bernoulli random variable with $\Ex(D) = \frac{5}{4} \cdot
m_u > m_u = \Ex(A)$. Since the departure process is faster than the arrival process, the stability of the queue is guaranteed by basic results in queueing theory \cite{Asmussen}. 
\end{proof}

We note that for linear power assignment, the dependence on $\Delta$
in Obs.~\ref{lem:signalincr} completely disappears, and thus also in Lemma \ref{lem:load} and
Thm.~\ref{th:lngthclass}.

\begin{corollary}
For all given networks with links on metric spaces,  
using linear power assignment,
\alg{Reflect} achieves an efficiency ratio that is an absolute constant (independent of $n$ and $\Delta$).
\end{corollary}

\noindent \textbf{Remark:}
In \alg{Reflect} we have assumed that each link $l_u$ knows $m_u$. In practice, this can be easily approximated at time $t$ by $\min\{1, \frac{A(t)}{t}\}$ (where $A(t)$ is the number
of packet arrivals up to time $t$), which
converges to the right value almost surely.

\subsection{Link partitioning}
If link lengths are known beforehand and some
pre-processing is allowed, the efficiency can be made
to have a better dependence on $\Delta$, specifically, we can achieve an efficiency ratio of $1/(6 \cdot 2^{\alpha} \kappa \log \Delta)$.

We can partition the link set into a collection of nearly equi-length
link sets, i.e., sets where the lengths in the set vary by at most a factor of $2$.
It is easy to show that a link set $L$ can be be partitioned into at most $\log_2 \Delta + 1$ sets of nearly equi-length links $L_r$ for $r = 1 \ldots \log_2 \Delta + 1$ where $L_r$ contains links of lengths in $[2^{r-1} \cdot \ell_{\min}, 2^r \cdot \ell_{\min})$.

We partition the time slots accordingly, a time slot $t$
is used to schedule links from class $L_r$ where $r = ((t - 1) \bmod{(\log_2 \Delta + 1)}) + 1$. 

With the partition, the arrival process on a sequence of slots devoted to a single length
class is equivalent to the setting where all links are nearly equi-length 
and assuming that $\sum_i m_i = 1/(6\cdot 2^{\alpha}\kappa)$.\footnote{There are certain technicalities here,
since a) $M_i \cap L_t$ may not be a \emph{maximal} feasible set in $L_t$ and b) Two sets $M_i$ and $M_j$ 
may have the same ``projection'' in $L_t$, i.e., it is possible that $M_i \cap L_t = M_j \cap L_t$ for $M_i \neq M_j$.
These can be handled in a straightforward way.}
This partitioning combined with Lemma \ref{lem:load}
 proves the claimed efficiency.

%
%

\section{Simulations}
\label{sec:simulations}

To see how the distributed algorithm \alg{Reflect} performs, we ran simulations on instances based on random topology.  
The problem instances were created by generating random links in a rectangle with side length $100$.  The length of the links were uniform random variables between $\ell_{\min}$ and $\ell_{\max}$, which we set as $1$ and $20$ respectively.  We generated random transmission requests for 100,000 time slots while running the algorithms.  Our focus was on the behavior of the maximum queue length, i.e., the largest number of waiting transmission requests over all the links, which was measured every 10,000 rounds. 

\begin{figure}
\begin{center}
\includegraphics[width=0.45\textwidth]{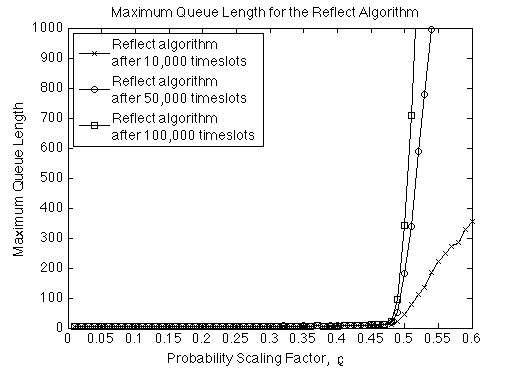}
\caption{The maximum queue lengths for the distributed algorithms Distr-SingleLink and Reflect.   The problem instances are based on random topology with $n = 200$, $\ell_{\min} = 1$, $\ell_{\max} = 20$, $\alpha = 2.5$ and $\beta = 1$.} \label{fig:maxqueue}
\end{center}
\end{figure}

Figure \ref{fig:maxqueue} shows the results of the distributed algorithm \alg{Reflect} for random instances with $200$ links after 10,000, 50,000 and 100,000 time slots.  The probability scaling factor, $\rho$, defines the load on the system.  Thus when $\rho = 1$, the system will, in expectation, receive a maximal feasible set in every round, while if the scaling factor is zero, no requests are generated.  The efficiency ratio of an algorithm is then equal to the largest $\rho$ such that the algorithm is stable.  We used a granularity of $0.01$ for values of $\rho$ between $0.01$ and $0.6$ and took the average over $10$ runs for each value of $\rho$.  

 The \alg{Reflect} algorithm used in Figure \ref{fig:maxqueue} approximates the arrival rate of requests for each link, as mentioned in an earlier remark, instead of assuming that the links know the arrival rate of requests.


It is interesting to note, in Figure \ref{fig:maxqueue}, that there is a very sharp threshold where the algorithm is no longer stable.  As soon as the algorithm becomes unstable, the queue lengths increase very rapidly.  As expected, the centralized Longest Queue First (LQF) algorithm was more stable than the distributed algorithm, managing to keep the maximum queue length below $2$ for all values of $\rho \leq 0.6$ and not becoming unstable until $\rho > 0.9$. We note here again that though our algorithm has a lower (though still high) throughput,
its main feature is the lack of a requirement for a centralized or even 
localized control. 

\newpage



\bibliographystyle{plain}
\bibliography{./references}		

\begin{thebibliography}{10}

\bibitem{Asmussen}
S{\o}ren Asmussen.
\newblock {\em Applied Probability and Queues}.
\newblock Springer, 2nd edition, 2003.

\bibitem{DBLP:journals/tit/ChaporkarKLS08}
Prasanna Chaporkar, Koushik Kar, Xiang Luo, and Saswati Sarkar.
\newblock Throughput and fairness guarantees through maximal scheduling in
  wireless networks.
\newblock {\em IEEE Transactions on Information Theory}, 54(2):572--594, 2008.

\bibitem{DBLP:conf/infocom/DaiP00}
J.~G. Dai and Balaji Prabhakar.
\newblock The throughput of data switches with and without speedup.
\newblock In {\em INFOCOM}, pages 556--564, 2000.

\bibitem{secordorder}
Antonis Dimakis and Jean Walrand.
\newblock Sufficient conditions for stability of longest-queue-first
  scheduling: second-order properties using fluid limits.
\newblock {\em Advances in Applied Probabability}, 38(2):505--521, 2006.

\bibitem{GHWW09}
Olga Goussevskaia, Magn\'{u}s~M. Halld\'{o}rsson, Roger Wattenhofer, and Emo
  Welzl.
\newblock {Capacity of Arbitrary Wireless Networks}.
\newblock In {\em INFOCOM}, pages 1872--1880, April 2009.

\bibitem{gouss2007}
Olga Goussevskaia, Yvonne~A. Oswald, and Roger Wattenhofer.
\newblock {Complexity in Geometric SINR}.
\newblock In {\em {Mobihoc}}, pages 100--109, 2007.

\bibitem{esatalg}
M.~M. Halld\'{o}rsson.
\newblock Wireless scheduling with power control.
\newblock {\em ACM Transactions on Algorithms}.
\newblock To appear. See also \url{http://arxiv.org/abs/1010.3427}, September
  2010.

\bibitem{SODA11}
Magn\'{u}s~M. Halld\'{o}rsson and Pradipta Mitra.
\newblock {Wireless Capacity with Oblivious Power in General Metrics}.
\newblock In {\em SODA}, 2011.

\bibitem{HW09}
Magn\'{u}s~M. Halld\'{o}rsson and Roger Wattenhofer.
\newblock {Wireless Communication is in APX}.
\newblock In {\em ICALP}, pages 525--536, July 2009.

\bibitem{bestInfocom08}
Changhee Joo, Xiaojun Lin, and N.B. Shroff.
\newblock { Understanding the Capacity Region of the Greedy Maximal Scheduling
  Algorithm in Multi-Hop Wireless Networks}.
\newblock In {\em INFOCOM}, 2008.

\bibitem{DBLP:journals/ton/JooS09}
Changhee Joo and Ness~B. Shroff.
\newblock Performance of random access scheduling schemes in multi-hop wireless
  networks.
\newblock {\em IEEE/ACM Trans. Netw.}, 17(5):1481--1493, 2009.

\bibitem{KesselheimSoda11}
Thomas Kesselheim.
\newblock {A Constant-Factor Approximation for Wireless Capacity Maximization
  with Power Control in the {S}{I}{N}{R} Model}.
\newblock In {\em SODA}, 2011.

\bibitem{KV10}
Thomas Kesselheim and Berthold V\"ocking.
\newblock Distributed contention resolution in wireless networks.
\newblock In {\em DISC}, pages 163--178, August 2010.

\bibitem{lqfmobihoc}
Long~B. Le, Eytan Modiano, Changhee Joo, and Ness~B. Shroff.
\newblock {Longest-queue-first scheduling under SINR interference model}.
\newblock In {\em MOBIHOC}, 2010.

\bibitem{lee09}
Hyang-Won Lee, Eytan Modiano, and Long~Bao Le.
\newblock Distributed throughput maximization in wireless networks via random
  power allocation.
\newblock In {\em IEEE Wiopt}, August 2009.

\bibitem{DBLP:conf/mobihoc/LiBX09}
Bo~Li, Cem Boyaci, and Ye~Xia.
\newblock A refined performance characterization of longest-queue-first policy
  in wireless networks.
\newblock In {\em MobiHoc}, pages 65--74, 2009.

\bibitem{MaheshwariJD08}
Ritesh Maheshwari, Shweta Jain, and Samir~R. Das.
\newblock A measurement study of interference modeling and scheduling in
  low-power wireless networks.
\newblock In {\em SenSys}, pages 141--154, 2008.

\bibitem{DBLP:conf/sigmetrics/ModianoSZ06}
Eytan Modiano, Devavrat Shah, and Gil Zussman.
\newblock Maximizing throughput in wireless networks via gossiping.
\newblock In Raymond~A. Marie, Peter~B. Key, and Evgenia Smirni, editors, {\em
  SIGMETRICS/Performance}, pages 27--38. ACM, 2006.

\bibitem{MoWa06}
T.~Moscibroda and R.~Wattenhofer.
\newblock {The Complexity of Connectivity in Wireless Networks}.
\newblock In {\em INFOCOM}, 2006.

\bibitem{Moscibroda2006Protocol}
Thomas Moscibroda, Roger Wattenhofer, and Yves Weber.
\newblock {Protocol Design Beyond Graph-Based Models}.
\newblock In {\em {Hotnets}}, November 2006.

\bibitem{mimosinr10}
D.~Qian, D.~Zheng, J.~Zhang, and Shroff N.
\newblock {CSMA}-based distributed scheduling in multi-hop {MIMO} networks
  under {SINR} model.
\newblock In {\em IEEE INFOCOM}, 2010.

\bibitem{shahfocs11}
Devavrat Shah, Jinwoo Shin, and Prasad Tetali.
\newblock Efficient distributed medium access (to appear).
\newblock In {\em FOCS}, 2011.

\bibitem{TE92}
L.~Tassiulas and A.~Ephremides.
\newblock Stability properties of constrained queueing systems and scheduling
  policies for maximum throughput in multihop radio networks.
\newblock {\em IEEE Trans.~Automat.~Contr.}, 37(12):1936--1948, 1992.

\end{thebibliography}

\end{document}